\documentclass[conference,10pt]{IEEEtran}

\usepackage{blindtext, graphicx,booktabs}
\usepackage{cite}
\usepackage{amsmath}
\usepackage{amssymb}
\usepackage{caption}
\usepackage[table,xcdraw]{xcolor}
\usepackage{subcaption}
\newtheorem{lemma}{Lemma}
\newtheorem{theorem}{Theorem}
\newcommand{\ignore}[1]{}
\usepackage{algorithm}
\usepackage{algpseudocode}

\begin{document}
\title{Optimal Vehicle Dimensioning for Multi-Class Autonomous Electric Mobility On-Demand Systems}
\author{\IEEEauthorblockN{Syrine Belakaria$^*$, Mustafa Ammous$^*$, Sameh Sorour$^*$ and Ahmed Abdel-Rahim$^\dag$$^\ddag$}
\IEEEauthorblockA{$^*$Department of Electrical and Computer Engineering,
University of Idaho,
Moscow, ID, USA\\ $^\dag$Department of Civil and Environmental Engineering,
University of Idaho,
Moscow, ID, USA\\$^\ddag$National Institute for Advanced Transportation Technologies, University of Idaho, Moscow, ID, USA\\
Email: \{ammo1375, bela7898\}@vandals.uidaho.edu, \{samehsorour, ahmed\}@uidaho.edu}
}

\maketitle
\begin{abstract}
Autonomous electric mobility on demand (AEMoD) has recently emerged as a cyber-physical system aiming to bring automation, electrification, and on-demand services for the future private transportation market. The expected massive demand on such services and its resulting insufficient charging time/resources prohibit the use of centralized management and full vehicle charging. A fog-based multi-class solution for these challenges was recently suggested, by enabling per-zone management and partial charging for different classes of AEMoD vehicles. This paper focuses on finding the optimal vehicle dimensioning for each zone of these systems in order to guarantee a bounded response time of its vehicles. Using a queuing model representing the multi-class charging and dispatching processes, we first derive the stability conditions and the number of system classes to guarantee the response time bound. Decisions on the proportions of each class vehicles to partially/fully charge, or directly serve customers are then optimized so as to minimize the vehicles in-flow to any given zone. Excess waiting times of customers in rare critical events, such as limited charging resources and/or limited vehicles availabilities, are also investigated. Results show the merits of our proposed model compared to other schemes and in usual and critical scenarios.
\end{abstract}
\IEEEpeerreviewmaketitle
\section{Introduction}
Urban transportation systems are facing tremendous challenges nowadays due to the dominant dependency and massive increases on private vehicle ownership, which result in dramatic increases in road congestion, parking demand \cite{ref5}, increased travel times \cite{ref6}, and carbon footprint \cite{ref3} \cite{ref4}. These challenges can be mitigated with the significant advances and gradual maturity of vehicle electrification, autonomous driving (10 million expected vehicles by 2020 \cite{ref11}), vehicle fast charging infrastructure, and most importantly cyber-physical systems capable of connecting all such components as well as customers to computing engines that can smartly exploit these resources. With the rapid development of such cyber-physical systems, it is strongly forecasted that vehicle ownership will significantly decline by 2025, and will be replaced by the novel concept Autonomous Electric Mobility on-Demand (AEMoD) services \cite{ref9,ref10}. In such system, customers will simply need to press some buttons on an app to promptly get an autonomous electric vehicle transporting them door-to-door, with no pick-up/drop-off and driving responsibilities, no dedicated parking needs, no carbon emission, no vehicle insurance and maintenance costs, and extra in-vehicle work/leisure times. AEMoD systems successfully exhibiting all these qualities will significantly prevail in attracting millions of subscribers across the world, and providing on-demand and hassle-free private urban mobility.\\ 
\indent Despite the great aspirations for wide AEMoD service deployments by early-to-mid next decade, the stability and timeliness (and thus success) of such service is threatened by two major bottlenecks. First, the expected massive demand of AEMoD services will result in excessive if not prohibitive computational and communication delays if cloud/centralized based approaches are employed for the micro-operation of such systems over an entire city. Moreover, the typical full-battery charging rates of electric vehicles will not be able to cope with the gigantic numbers of vehicles involved in these systems, thus resulting in instabilities and unbounded customer delays. Recent works \cite{ref1,ref12} have addressed some important problems related to autonomous mobility on-demand systems but none of them considered the computational architecture for a massive demand on such services, nor the vehicle electrification and charging limitations. \\
\indent In our prior work \cite{ref13}, a fog-based architecture \cite{ref14} with multi-class operation, and possible partial charging was proposed to handle these two problems. The fog-based architecture distributes the micro-management of AEMoD vehicles to zone controllers that are close to customers and their most likely allocated vehicles, thus reducing communications and computation loads and delays. The per-zone multi-class operation with partial charging pairs customers with vehicles either having the proper charge or needing a partial charge of their batteries to fulfill the customers' requested trips. The number of classes is selected to balance the proportions between the customer demands, vehicle in-flow, and charging resources of each zone. Decisions on the proportions of vehicles of each class to dispatch or partially/fully charge were optimized in this work to minimize the \emph{response time} of the system.\\
\indent While the proposed architecture, multi-class operation, and joint dispatching and charging optimization framework in \cite{ref13} seems very promising, the study assumed a constant vehicle in-flow to each zone. Though this typical by the active vehicle in-flow to the system (in-flow of vehicles dropping customers in this zone), the zone demand may require more (less) vehicles at any given time of the day, which may call for relocating excess vehicles from (to) neighboring zones. One one hand, serving customers within bounded response times can be guaranteed by injecting more vehicles to each zone. On the other hand, one of the key goals of AEMoD systems is to reduce the congestion. Therefore, determining the optimal number of needed vehicles (a.k.a vehicle dimensioning) to stably serve each zone with bounded response time guarantees is very crucial factor in the operation and key goals of AEMoD systems. In addition, such systems need to be resilient and maintain their stability in special conditions like low charging resources, limited vehicles availability, etc.\\
\indent In this paper, we address the above vehicle dimensioning problem with bounded response time guarantees for the fog-based multi-class AEMoD management system proposed in \cite{ref13}. Using a queuing model representing the multi-class charging and dispatching processes of each zone, we first derive the stability conditions and the number of system classes to guarantee the response time bound. Decisions on the proportions of each class vehicles to partially/fully charge, or directly serve customers are then optimized so as to minimize total needed vehicles in-flow to any given zone. Excess waiting times of customers in rare critical events, such as limited charging resources and/or limited vehicles availabilities, will be also investigated.

\section{System Model and Parameters} \label{sec:model}
\indent We consider one service zone controlled by a fog controller connected to: (1) the service request apps of customers in the zone; (2) the AEMoD vehicles; (3) $C$ rapid charging poles distributed in the service zone and designed for short-term partial charging; and (4) one spacious central charging station designed for long-term full charging. Active AEMoD vehicles enter the service in this zone after dropping off their latest customers in it. Their detection as free vehicles by the zone's controller can thus be modeled as a Poisson process with rate $\lambda_v $. Customers request service from the system according to a Poisson process. Both customers and vehicles are classified into $n$ classes based on an ascending order of their required trip distance and the corresponding suitable SoC for this trip, respectively. From the thinning property of Poisson processes, the arrival process of Class $i$ customers and vehicles, $i \in \{0,\dots, n\}$, are both independent Poisson processes with rates $\lambda_c^{(i)}$ and $\lambda_v p_i$, where $p_i$ is the probability that the SoC of an arriving vehicle to the system belongs to Class $i$. Note that $p_0$ is the probability that a vehicle arrives with a depleted battery, and is thus not able to serve customers immediately. Consequently, $\lambda^{(0)}_c = 0$ as no customer will request a vehicle that cannot travel any distance. On the other hand, $p_n$ is also equal to 0, because no active vehicle can arrive to the system fully charged as it has just finished a prior trip.\\
\ignore{
\begin{equation}\label{eq:1}
\begin{aligned}
\sum ^{n-1}_{i=0}p_i = 1 , \; 0 \leq p_{i} \leq 1 , \; i = 0, \ldots, n-1.
\end{aligned}
\end{equation}
}
\begin{figure}[t]
\centering
 \includegraphics[width=.55\textwidth]{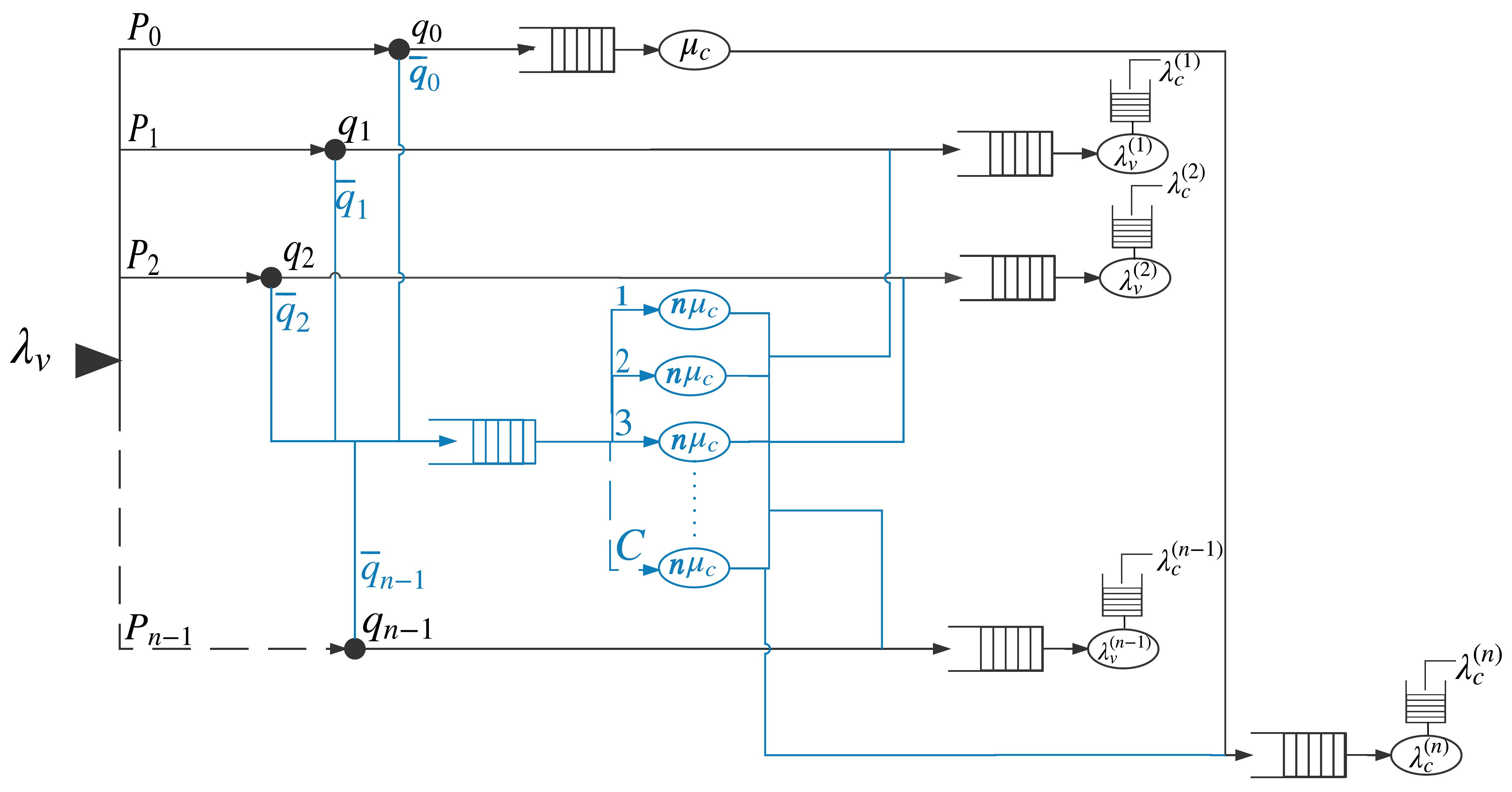}
\caption{Joint dispatching and partially/fully charging model, abstracting one service zone of an AEMoD system.}\label{fig:model}
   \end{figure}
\indent Upon arrival, each vehicle of Class $i$, $i\in\{1,\dots,n-1\}$, will park anywhere in the zone until it is directed by the fog controller to either: (1) dispatch to serve a customer from Class $i$ with probability $q_i$; or (2) partially charge up to the SoC of Class $i+1$ at any of the $C$ charging poles (whenever any of them becomes free), with probability $\overline{q}_i=1-q_i$, before parking again in waiting to serve a customer from Class $i+1$. As for Class $0$ vehicles that are incapable of serving before charging,\ignore{ since they cannot serve any customer class without charging,} they will be directed to either fully charge at the central charging station with probability $q_0$, or partially charge at one of the $C$ charging points with probability $\overline{q}_0 = 1-q_0$. In the former and latter cases, the vehicle after charging will wait to serve customers of Class $n$ and $1$, respectively.\\
\ignore{
\begin{equation}\label{eq:2}
\begin{aligned}
q_i  + \overline{q_i} =1 , \; 0 \leq q_{i} \leq 1  , \; i = 0, \ldots, n-1.
\end{aligned}
\end{equation}
}
\indent As widely used in the literature (e.g., \cite{ref15,ref16}), the full charging time of a vehicle with a depleted battery is assumed to be exponentially distributed with rate $\mu_c$, to model the random charging duration of different battery sizes. Given a uniform SoC quantization among the $n$ vehicle classes, the partial charging time can then be modeled as an exponential random variable with rate $n\mu_c$. Note that the larger rate of the partial charging process is not due to a speed-up in the charging process but rather due to the reduced time of partial charging.\ignore{ The use of exponentially distributed charging times for charging electric vehicles has been widely used in the literature \cite{ref15,ref16} to model the randomness in the charging duration of the different battery sizes.} The customers belonging to Class $i$, arriving at rate $\lambda_c^{(i)}$, will be served at a rate of $\lambda_v^{(i)}$, which includes the vehicle in-flow that: (1) arrived to the zone with a SoC belonging to Class $i$ and were directed to wait to serve Class $i$ customers; or (2) arrived to the zone with a SoC belonging to Class $i-1$ and were directed to partially charge to be able to serve Class $i$ customers.\\
\indent Given the above description and modeling of variables\ignore{ of the vehicle dispatching and charging variables and options}, the entire zone dynamics can thus be modeled by the queuing system depicted in Fig.\ref{fig:model}. This system includes $n$ M/M/1 queues for the $n$ classes of customer service, one M/M/1 queue for the central charging station, and one M/M/C queue representing the partial charging process at the $C$ charging points.\\
\indent Our goal in this paper is to minimize the needed rate of vehicle in-flow $\lambda_v$ to the entire zone with respect to the arrival rate of customers in order to guarantee an average response time limit for customers of every class. By response time, we mean the time elapsed between the instant when an arbitrary customer requests a vehicle, and the instant when a vehicle starts moving from its parking or charging spot towards this customer. We will also shade light on the potential dimensioning and/or response time relaxation solutions for system resilience in extreme cases of very low energy resources and limited actual vehicle in-flow.

\section{System Stability and Response Time Limit Conditions} \label{sec:stability}
In this section, we first deduce the stability conditions of the proposed system using the basic laws of queuing theory. We will also derive a lower bound on the number of classes $n$ that fits the customer demands, average response time limit, and charging capabilities of any arbitrary service zone.
As shown in Fig. \ref{fig:model}, each of the $n$ customer classes is served by a separate queue of vehicles having a vehicle in-flow rate $ \lambda_v^{(i)}$. Consequently, $ \lambda_v^{(i)}$ represents the service rate of the customer arrival in the $i^{th}$ queue.
From the aforementioned vehicle dispatching and charging dynamics in Section \ref{sec:model}, illustrated in Fig. \ref{fig:model}, these service rates can be expressed as:
\begin{equation}\label{eq:1}
\begin{aligned}
& & &\lambda_v^{(i)} = \lambda_v(p_{i-1}\overline{q}_{i-1} + p_{i}q_{i}) , \; i = 1, \ldots, n-1.\\
& & &\lambda_v^{(n)} = \lambda_v(p_{n-1}\overline{q}_{n-1} + p_{0}q_{0})\\
\end{aligned}
\end{equation}
Since $\overline{q}_{i} + q_{i} = 1$, $\overline{q}_{i}$ can be substituted by  $1 - q_{i}$ to get:
\begin{equation}\label{eq:2}
\begin{aligned}
& & &\lambda_v^{(i)} = \lambda_v(p_{i-1} - p_{i-1}{q_{i-1}} + p_{i}q_{i}) , \; i = 1, \ldots, n-1\\
& & &\lambda_v^{(n)} = \lambda_v(p_{n-1} - p_{n-1}{q_{n-1}} + p_{0}q_{0}) \ignore{, \; i = n}\\
\end{aligned}
\end{equation}
From the well-known stability condition of an M/M/1 queue, we must have: 
\begin{equation}\label{eq:3}
\begin{aligned}
& & & \lambda_v^{(i)} > \lambda_c^{(i)} , \; i = 1, \ldots, n 
\end{aligned}
\end{equation}
It is also established from M/M/1 queue analysis that the average response (i.e., service) time for any customer in the $i$-th class can be expressed a:
\begin{equation} \label{eq:5}
\dfrac {1}{\lambda_v^{(i)}-\lambda_c^{(i)}}
\end{equation}
To guarantee customers' satisfaction, the fog controller of each zone must impose an average response time limit $T$ for any class. We can thus express this average response time constraint for the customers of the $i$-th class as:
\begin{equation} \label{eq:6}
\dfrac {1}{\lambda_v^{(i)}-\lambda_c^{(i)}} \leq T
\end{equation}
which can also be re-written as:
\begin{equation} \label{eq:7}
\lambda_v^{(i)} - \lambda_c^{(i)}\geq \frac{1}{T}
\end{equation}
\indent Before reaching the customer service queues, the vehicles will go through a decision step of either going to these queues immediately or partially charging. The stability of the charging queues should be guaranteed in order to ensure the global stability of the entire system at the steady state. From the model described in the previous section, and the well-known stability conditions of M/M/C and M/M/1 queues, we get the following stability constraints on the $C$ charging points and one central charging station queues, respectively: 
\begin{equation}\label{eq:4}
\begin{aligned}
& & & \sum ^{n-1}_{i=0}\lambda_v(p_{i} - p_{i}{q_{i}}) < C (n \mu_c) \\
& & &\lambda_v p_{0}q_{0} < \mu_c
\end{aligned}
\end{equation}
\indent The following lemma sets a lower bound on the average vehicle in-flow rate to the entire service zone to guarantee both its stability and the average response time limit fulfillment for all its classes, given their demand rates.
\begin{lemma}\label{lem2}
For the entire zone stability, and fulfillment of the average response time limit for all its classes, the average vehicles in-flow rate must be lower bounded by:
 \begin{equation}\label{eq:9}
\begin{aligned}
\lambda_v \geq \sum ^{n}_{i=1}{\lambda_c^{(i)}} + \frac{n}{T}
\end{aligned}
\end{equation}
\end{lemma}
\begin{proof}
The proof of Lemma \ref{lem2} is in Appendix A in \cite{ref18}.
\end{proof}
Furthermore, the following lemma establishes a lower bound on the number of classes $n$ that fits zone's customer demands, average response time limit, and charging capabilities. 
\begin{lemma}\label{lem3}
For stablize the zone operation given its customer demands, average response time limit, and charging capabilities, the number of classes $n$ in the zone must obey the following inequality:
\begin{equation}\label{eq:11}
\begin{aligned}
n  \geq  \dfrac{\sum ^{n}_{i=1}{\lambda_c^{(i)}} - \mu_c}{C \mu_c - 1/T}
\end{aligned}
\end{equation}
\end{lemma}
\begin{proof}
The proof of Lemma \ref{lem3} is in Appendix B in \cite{ref18}.
\end{proof}

\section{Optimal Vehicle Dimensioning}
\subsection{Problem Formulation}
As previously mentioned, this paper aims to minimize the average vehicle in-flow rate $\lambda_v$ to the entire zone, given its charging capacity and customer demand rates, while guaranteeing an average response time limit for each class customers. Given the described system dynamics in Section \ref{sec:model} and the derived conditions in Section \ref{sec:stability}, the above problem can be formulated as a stochastic optimization problem as follows:
\begin{subequations}\label{eq:12}
\begin{align}
&\qquad\qquad \underset{q_0,q_1,\ldots ,q_{n-1}}{\text{minimize }} \lambda_v \\
\text{s.t}  & \nonumber\\
&\lambda_c^{(i)} - \lambda_v(p_{i-1} - p_{i-1}{q_{i-1}} + p_{i}q_{i}) +\frac{1}{T} \leq 0, \; i = 1, \ldots, n-1 \label{eq:12-C1}\\
& \lambda^{(n)}_c - \lambda_v(p_{n-1} - p_{n-1}{q_{n-1}} + p_{0}q_{0}) +\frac{1}{T}\leq 0 \ignore{,\; i = n} \label{eq:12-C2}\\
&\sum ^{n-1}_{i=0}\lambda_v(p_{i} - p_{i}{q_{i}}) - C (n \mu_c) < 0\label{eq:12-C3}\\
&\lambda_v p_{0}q_{0} - \mu_c<  0 \label{eq:12-C4} \\
& \sum ^{n-1}_{i=0}p_i = 1 , \; 0 \leq  p_{i} \leq  1  , \; i = 0, \ldots, n-1  \label{eq:12-C5}\\
&0 \leq q_{i} \leq 1  , \; i = 0, \ldots, n-1  \label{eq:12-C6}\\
&\lambda_v \geq \sum ^{n}_{i=1}{\lambda_c^{(i)}} + \frac{n}{T}  \label{eq:12-C7}
\end{align}
\end{subequations}
The $n$ constraints in (\ref{eq:12-C1}) and (\ref{eq:12-C2}) represent the stability and response time limit conditions of the system introduced in (\ref{eq:7}), after substituting every $\lambda_v^{(i)}$ by its expansion form in (\ref{eq:2}). The constraints in (\ref{eq:12-C3}) and (\ref{eq:12-C4}) represent the stability conditions for the charging queues. The constraints in (\ref{eq:12-C5}) and (\ref{eq:12-C6}) are the axiomatic constraints on probabilities (i.e., values being between 0 and 1, and sum equal to 1). Finally, Constraint (\ref{eq:12-C7}) is the lower bound on $\lambda_v$ introduced by Lemma (\ref{lem2}).\\
The above optimization problem is a quadratic non-convex problem with second order differentiable objective and constraint functions. Usually, the solution obtained by using the Lagrangian and KKT analysis for such non-convex problems provides a lower bound on the actual optimal solution. Consequently, we propose to solve the above problem by first finding the solution derived through Lagrangian and KKT analysis, then, if needed, iteratively tightening this solution to the feasibility set of the original problem.

\subsection{Lower Bound Solution}
\ignore{As stated above, applying the KKT conditions to the constraints of the problem and the gradient of the Lagrangian function allows us to find an analytical solutions on the decisions variables $q_i$ achieving a lower bound on the optimal vehicle in-flow rate $\lambda_v$. This section aims to find this lower bound solution.}The Lagrangian function associated with the optimization problem in (\ref{eq:12}) is given by the following expression:
\begin{multline}\label{eq:13}
\begin{aligned}
& L(\mathbf{q},\lambda_v,\boldsymbol{\alpha},\boldsymbol{\beta},\boldsymbol{\gamma},\boldsymbol{\omega}) =
 \lambda_v +\alpha_{n} ( \lambda_v( p_{n-1}q_{n-1} - p_{0}{q_{0}} \\ 
&-p_{n-1} )+ \lambda_c^{(n)} + \frac{1}{T} ) +\sum_{i=1}^{n-1}\alpha_{i}( \lambda_v(p_{i-1}{q_{i-1}} -  p_{i}q_{i} -p_{i-1})
 \\
& + \lambda_c^{(i)} +\frac{1}{T} ) + \beta_{0} (\sum ^{n-1}_{i=0}\lambda_v(p_{i} - p_{i}{q_{i}}) - C n \mu_c + \epsilon_0 ) \\ 
& + \beta_{1}(\lambda_v p_{0}q_{0} - \mu_c + \epsilon_1 )+ \sum_{i=0}^{n-1} \gamma_{i}(q_{i} - 1) - \sum_{i=0}^{n-1}\omega_{i}q_{i}\\
&- \omega_{n}(\lambda_v -  \sum ^{n}_{i=1}{\lambda_c^{(i)}} - \frac{n}{T})\;,
\end{aligned}
\end{multline}
where $\mathbf{q}$ is the vector of dispatching decisions (i.e. $\mathbf{q} = [q_0,\dots,q_{n-1}]$), and:
\begin{itemize}
\item $\boldsymbol{\alpha} = [\alpha_{i}]$, such that $\alpha_{i}$ is the associated Lagrange multiplier to the $i$-th customer queue inequality.
\item $\boldsymbol{\beta} = [\beta_{i}]$, such that $\beta{i}$ is the associated Lagrange multiplier to the $i$-th charging queue inequality.
\item $\boldsymbol{\gamma} = [\gamma_{i}]$, such that $\gamma_{i}$ is the associated Lagrange multiplier to the $i$-th upper bound inequality on $q_i$.
\item $\boldsymbol{\omega} = [\omega_{i}]$, such that $\omega_{i}$ is the  associated Lagrange multiplier to the $i$-th lower bound inequality on $q_i$ and $\lambda_v$.
\end{itemize}
For more accurate resolutions, two small positive constants $\epsilon_0$ and $\epsilon_1$ are added to the stability conditions on the charging queues to make them non strict inequalities.\\
\indent Solving the equations given by the KKT conditions on the problem equality and inequality constraints, the following theorem illustrates the optimal lower bound solutions of the problem in (\ref{eq:12}) 
\begin{theorem}\label{thm1}
The lower bound solution of the optimization problem in (\ref{eq:12}), obtained from Lagrangian and KKT analysis can be expressed as follows:
\begin{equation}\label{eq:15}
\begin{aligned}
& \lambda_v^* = \begin{cases}
           \sum ^{n}_{i=1}{\lambda_c^{(i)}} + \frac{n}{T} ~~~  \omega_n^* \ne 0 \\    
            \sum_{i=1}^{n}\alpha_{i}^*(\lambda_c^{(i)}+\frac{1}{T}) -\beta_0^* ( C n \mu_c - \epsilon_0) - \beta_1^* (\mu_c -\epsilon_1) \\ ~~~~~~~~~~~~~~~~~~~~~~~~~~~~~~~~~~~~~~~~~~~~~~~~~~~~~~ \omega_n^* = 0\\ 
            \end{cases} \\ 
& q_{0}^* = \begin{cases}
            0 &  \alpha_{1}^* - \alpha_{n}^* - \beta_0^* + \beta_1^*>0 \\
            1 & \alpha_{1}^* - \alpha_{n}^* - \beta_0^* + \beta_1^* < 0 \\
            \frac{p_{n-1}{q_{n-1}^* - p_{n-1} }}{p_{0}} + \frac{ \lambda_c^{(n)} + \frac{1}{T}}{\lambda_v p_{0} } & \alpha_{n}^* \ne 0 \\
            \frac{\mu_c}{\lambda_v^* p_0}& \beta_1^* \ne 0 \\
            \zeta_0(\alpha^*,\beta^*,\gamma^*,\lambda_v^*,q^*)& Otherwise \\
            \end{cases} \\ 
& q_{i}^* =   \begin{cases}
            0 &  \alpha_{i+1}^* - \alpha_{i}^* - \beta_0^* > 0 \\
            1 &  \alpha_{i+1}^* - \alpha_{i}^* - \beta_0^* < 0  \\
            \frac{p_{i-1}{q_{i-1}^*} -p_{i-1}}{p_{i}} + \frac{\lambda_c^{(i)} + \frac{1}{T}}{\lambda_v p_{i} } &\alpha_{i}^* \ne 0\\
             \zeta_i(\alpha^*,\beta^*,\gamma^*,\lambda_v^*,q^*)& Otherwise \\
            \end{cases}  \\      
& i= 1, \ldots, n-1.
\end{aligned}
\end{equation}
where $\zeta_i(\alpha^*,\beta^*,\gamma^*,\lambda_v^*,q^*)$ is the solution that that maximize $\underset{\mathbf{q}}\inf~ L(\mathbf{q},\alpha^*,\beta^*,\gamma^*,\lambda_v^*)$
\end{theorem}
\begin{proof}
The proof of Theorem \ref{thm1} is in Appendix C in \cite{ref18}.
\end{proof}
\subsection{Solution Tightening}
As stated earlier, the closed-form solution derived in the previous section from analyzing the constraints' KKT conditions does not always match with the optimal solution of the original optimization problem, and is sometimes a non-feasible lower bound on our problem. Unfortunately, there is no method to find the exact closed-from solution of non-convex optimization. However, starting from the derived lower bound, we can numerically tighten this solution by toward the feasibility set of the original problem. There are several algorithms to iteratively tighten lower bound solutions, one of which is the \textit{Suggest-and-Improve algorithm} algorithm proposed in \cite{qcqp} to tighten non-convex quadratic problems.\ignore{\\
\begin{algorithm}
\caption{Suggest-and-Improve algorithm}
\begin{algorithmic}
\State{1. Suggest. Find a candidate point $q \in D^n$ and $\lambda_v \in R$.\\
2. Improve. Run a local method from q to find a point $z \in D^n$ and $x \in R$ that is no worse than q and $\lambda_v$.\\
3. return z and x.}
\end{algorithmic}
\end{algorithm}\\
with q is the vector of the decisions $q_i$ and D=[0,1] is the domain in which each $q_i$ is defined.} We will thus propose to employ this method whenever the KKT conditions based solution is not feasible and tightening is required.

\section{Simulation Results}
In this section, we test both the performance and merits of the proposed dimensioning solution for the considered multi-class AEMoD system. The metric of interest in this study is the optimal vehicle in-flow rate to an arbitrary zone of interest. The performance of the proposed dimensioning solution is tested for two possible SoC distributions for in-flow vehicles, namely the decreasing and Gaussian distributions. The former distribution better models the more probable active-vehicle-dominant in-flow scenarios, as such vehicles typically exhibit higher chances of having lower battery charge. The latter distribution models the rarer relocated-vehicle-dominant in-flow scenarios, as such vehicles typically charge for random amounts of times before relocating to the zone of interest. Customers trip distances are always assumed to follow a Gaussian distribution because customers requiring mid-size distances are usually more than those requiring very small and very long distances. For all the performed simulation studies, the full-charging rate of a vehicle is set to $\mu_c = 0.033$ mins$^{-1}$. Moreover, for Figures 2, 3, and 4, the number of charging poles $C$ is set to 40.\\
\indent The first important finding of this study is that the obtained solutions using the closed-form expressions in Theorem 1 (i.e., the one derived by applying the KKT conditions) were always feasible solutions to the original problem in (\ref{eq:12}), for the entire broad range of system parameters employed in our simulations. Thus, the derived closed-form solution is in fact the optimal dimensioning solution for a broad range of system settings, and no tightening is needed.\\
\indent Fig. \ref{fig:tradeoff} shows the trade-off relation between the average response time limit, total customer demand rate, and the optimal vehicle in-flow rate, for both vehicle SoC distributions. This curve can be used by the fog controller to get a rough estimate (without exact demand information per class nor optimization of the dispatching and charging dynamics) on its required in-flow rate (and thus whether it needs extra vehicles or have excess vehicles to relocate) for any given customer demand rate and desired response time limit.\ignore{ It can be also employed to calculate and notify customers about their expected response time limit for scenarios where the vehicle in-flow rate cannot be changed.}\\
\begin{figure}[t]
\centering
  \includegraphics[width=.98\linewidth]{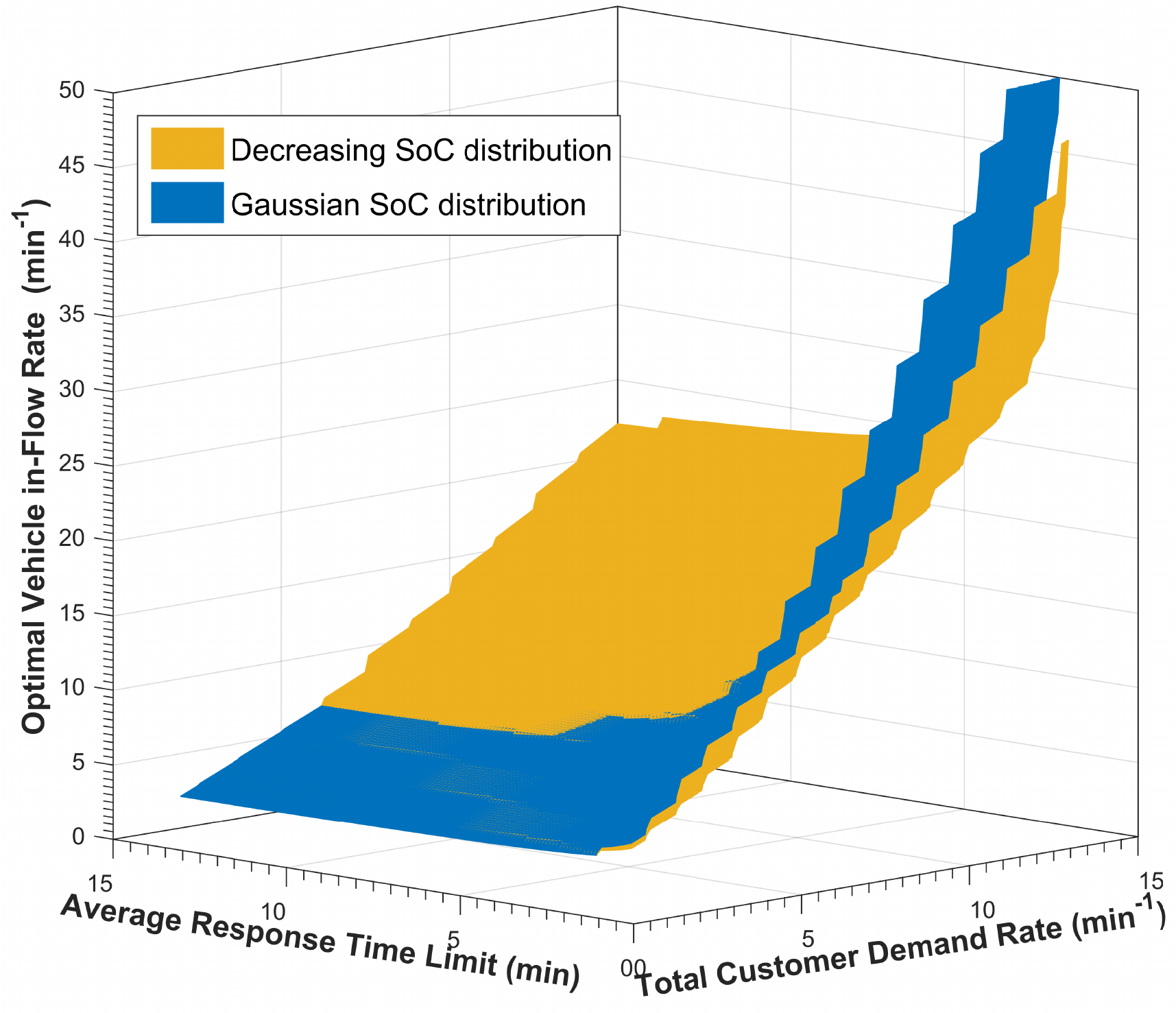}
    \caption{Effect of varying the average response time limit and total customer demand rate.}
 \label{fig:tradeoff}
\end{figure}
\indent Fig. \ref{fig:merged} illustrates the effect of increasing the number of classes $n$ beyond its lower bound introduced in Lemma \ref{lem3} for both variable total customer demand rate (while fixing the average response time limits to 5 mins) and variable average response time limits (while fixing the total customer demand rate to 5 min$^{-1}$) in the left and right sub figures, respectively. Both decreasing and Gaussian SoC distributions are considered. In both sub-figure, the lower bound on the number of classes vary depending on the values of the average response time and the total customer demand rate (as shown in Lemma 2), with maximum values of 14 and 11 for the employed values in the left and right sub-figures, respectively. The results in both figures clearly show that increasing $n$ beyond its lower bound increases the required vehicle in-flow to the zone. We thus conclude that the optimal number of classes is the smallest integer value satisfying Lemma \ref{lem3}.\\ \ignore{
\begin {small}
\begin{equation}
n^* = \begin{cases}  \dfrac{\sum ^{n}_{i=1}{\lambda_c^{(i)}} - \mu_c}{C \mu_c - 1/T} &  \dfrac{\sum ^{n}_{i=1}{\lambda_c^{(i)}} - \mu_c}{C \mu_c - 1/T} ~\text{integer}\\
\left\lceil \dfrac{\sum ^{n}_{i=1}{\lambda_c^{(i)}} - \mu_c}{C \mu_c - 1/T}\right\rceil & \text{Otherwise}
\end{cases}
\end{equation}
\end{small} }
\begin{figure}[t]
\centering
  \includegraphics[width=.98\linewidth]{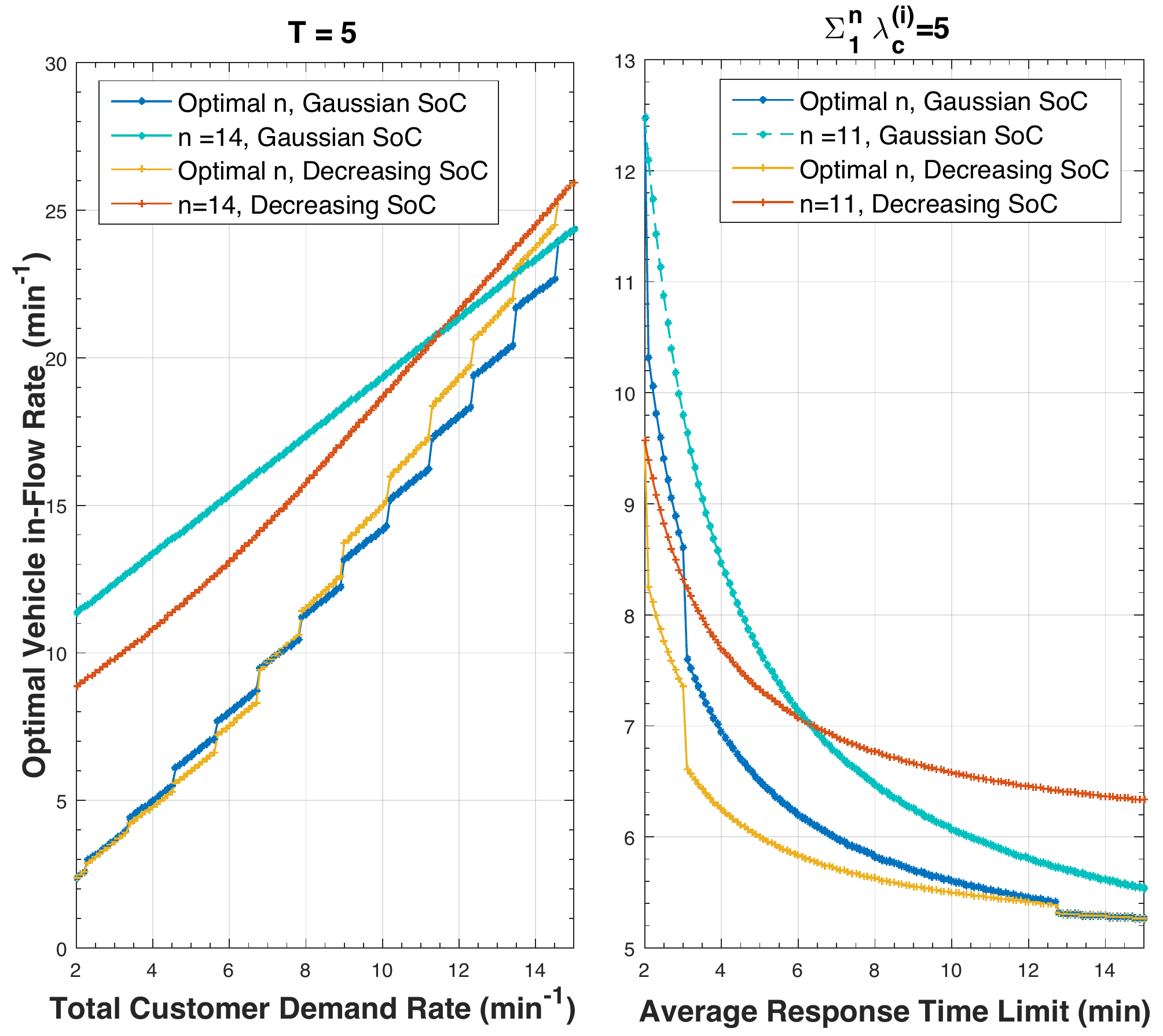}
    \caption{Effect of increasing the number of classes.}
 \label{fig:merged}
\end{figure}
\indent Fig. \ref{fig:mergedspecialcases} compares the performance of our proposed optimal vehicle dimensioning scheme with other non-optimized approaches (in which vehicles follow a fixed dispatching/charging policy irrespective of the system parameters) for different values of total customer demand rate (with $T=5$) and average response time limit (with $\sum_{i=1}^n\lambda_c^{(i)}=5$). The two non-optimized approaches are the always-charge approach (i.e. $q_i=0~\forall~i$) and the equal-split approach (i.e. $q_i=0.5~\forall~i$). The figure clearly shows the superior performance of our derived optimal policy compared to the two non-optimized policies, especially for large total customer demand rates and lower average response time limits. For $\sum ^{n}_{i=1}{\lambda_c^{(i)}} = 10$ min$^{-1}$ in the left subfigure, 36\% and 44.4\% less vehicle in-flow rates are required compared to always-charge and equal-split policies, respectively, for the more typical decreasing SoC distribution. These reductions reach 57.6\% and 42.5\%, respectively for $T = 10$ min in the right subfigure. The always-charge policy is exhibiting less increase in the required vehicle in-flow rate when the SoC follows a Gaussian distribution. However, some considerable gains can still be achieved using our proposed optimized approach in this less frequent SoC distribution setting. Noting that these gains can be higher in more critical scenarios, the results demonstrate the importance of our proposed scheme in establishing a better engineered and more stable system with less vehicles.\\
\begin{figure}[t]
\centering
  \includegraphics[width=.98\linewidth]{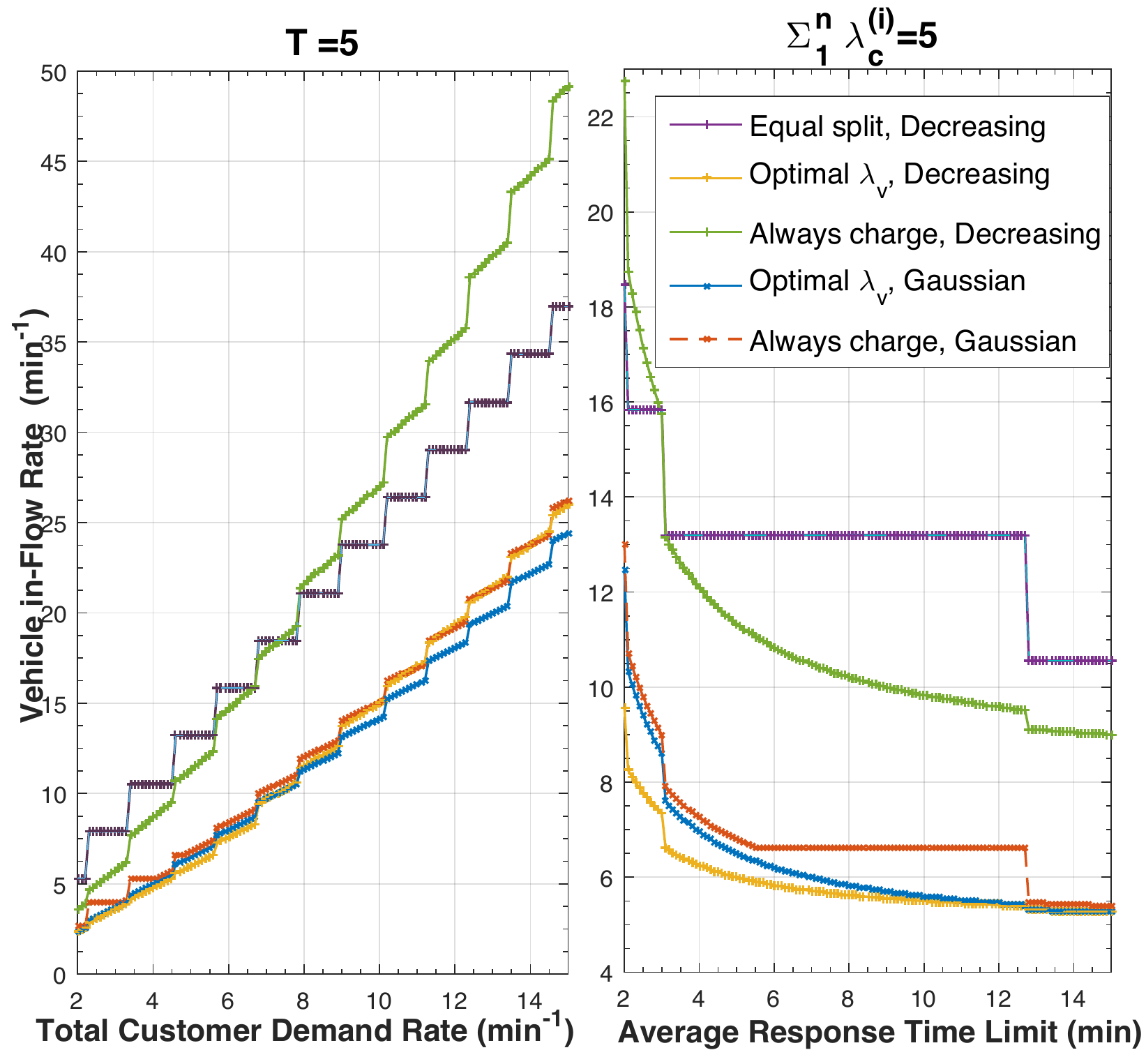}
    \caption{Comparison to non-optimized policies.}
 \label{fig:mergedspecialcases}
\end{figure}
\indent Finally, we studied the resilience requirements for our considered model in the critical scenarios of sudden reduction in the number of charging sources within the zone. This reduction may occur due to either natural (e.g., typical failures of one or more stations) or intentional (e.g., a malicious attack on the fog controller blocking its access to these sources).\ignore{ It is clear that the stable operation of AEMoD systems strongly depends of the power energy resources represented by the number of available charging poles $C$ in each zone. In order for the system to be resilient, it must be able to quickly adapt to and/or recover from sudden failures of one or more charging stations.} The resilience measures that the fog controller can take in these scenarios is to both notify its customers of a transient increase in the vehicles' response times given the available vehicles in the zone, and request a higher vehicle in-flow rate to gradually restore its original response time limit.\\ 
\indent Our developed optimization framework in \cite{ref13} and this paper can easily provide proper numbers for both the above two needed actions by the fog controller in charging station outage events. The problem of computing the maximum transient response time of the system given the fixed vehicle in-flow rate at failure time was already solved in our previous related work \cite{ref13}. The left subfigure of Fig. \ref{fig:mergedC} depicts the maximum response time values of the system for different numbers of available charging poles for a vehicle in-flow rate $\lambda_v = 8$ min$^{-1}$ and a total customer demand rate of $5$ min$^{-1}$. For a Gaussian distribution of vehicles' SoC, the response time increases dramatically when the number of charging poles drops below 20. On the other hand, the degradation in response time was much less severe when the SoC of vehicles follows the decreasing distribution. Luckily, the decreasing SoC distribution is the one that is more probable especially at the time just preceding the failure (where most vehicles arriving to the zone are active vehicles).\\
\indent As for the recovery from this critical scenario and restoration of the original response time limit, the proposed dimensioning framework in this paper can be employed to determine the new optimal value of vehicle in-flow rate. The right sub-figure in Fig. \ref{fig:mergedC} depicts the optimal vehicles in-flow $\lambda_v^*$ for different values of available charging poles $C$. In this simulation, the total customer demand rate is set to $\sum ^{n}_{i=1}{\lambda_c^{(i)}} = 8$ min$^{-1}$ and the average response time limit is restored back to $T=10$ mins. The figure shows that the Gaussian SoC distribution case, which would be luckily the dominant case in this zone after failure time (due to the domination of relocated vehicles called in by the fog controller to recover from the failure event), exhibit lower need of vehicle in-flow rate to restore the system conventional operation.    

\section{Conclusion}
This paper aimed to formally characterize the optimal vehicle dimensioning for fog-based multi-class AEMoD systems given a system-wide average response time limit. Using the system's queuing model and its stability/response-time constraints, we formulated the optimal vehicle dimensioning problem as a non-convex quadratic program over the multi-class dispatching and charging proportions. The lower bound solution corresponding to the Lagrangian and KKT-conditions analysis of the problem were analytically derived, and were shown to match the optimal solution of the original problem for a broad range of system parameters using extensive simulations. The optimal number of classes to minimize the required vehicle in-flow rate was also characterized. Simulation results demonstrated the merits of our proposed optimal decision scheme compared to other schemes. They also illustrated the resilience requirements calculated using our proposed solutions to recover from sudden reductions in charging resources.

\begin{figure}[t]
\centering
  \includegraphics[width=.98\linewidth]{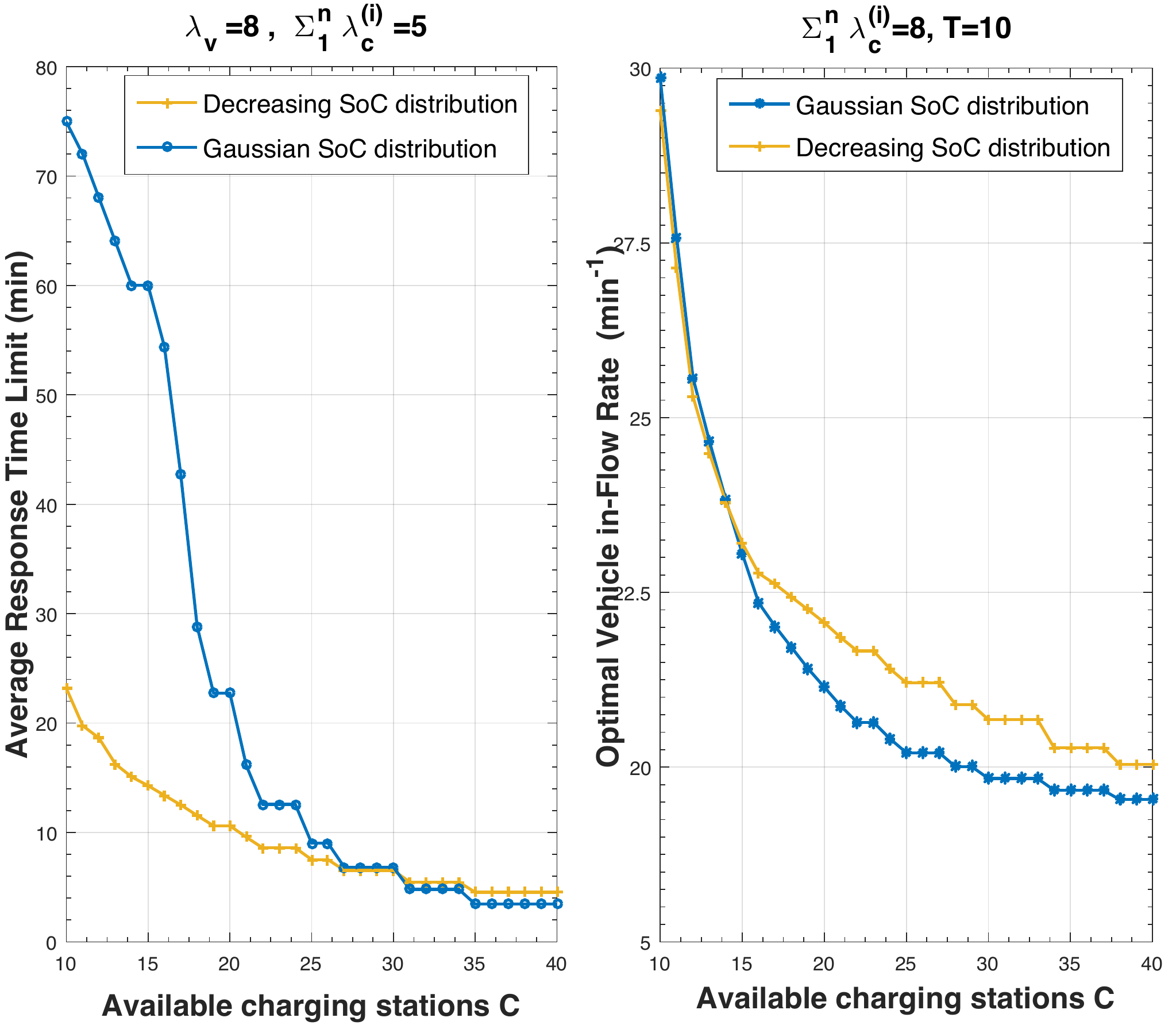}
    \caption{Effect of varying the charging point availability.}
 \label{fig:mergedC}
\end{figure}

\clearpage
\appendices
\section{Proof of Lemma \ref{lem2}}\label{app:lem2}
From (\ref{eq:2}) and (\ref{eq:7}) we have 
\begin{equation}\label{eq:22}
\begin{aligned}
& & & \lambda_c^{(i)} + \frac{1}{T} \leq \lambda_v(p_{i-1}\overline{q}_{i-1} + p_{i}q_{i}) , \; i = 1, \ldots, n-1.\\
& & & \lambda_c^{(n)} + \frac{1}{T} \leq\lambda_v(p_{n-1}\overline{q}_{n-1} + p_{0}q_{0}),\; i = n
\end{aligned}
\end{equation}

The summation of all the inequalities in (\ref{eq:22}) gives a new inequality 
\begin{equation}\label{eq:23}
\begin{aligned}
\sum ^{n}_{i=1}{\lambda_c^{(i)}} + \frac{n}{T} \leq \lambda_v [\sum ^{n-1}_{i=1}{(p_{i-1}\overline{q}_{i-1} + p_{i}q_{i})} + (p_{n-1}\overline{q}_{n-1} + p_{0}q_{0})]
\end{aligned}
\end{equation}
\begin{equation}\label{eq:24}
\begin{aligned}
\sum ^{n}_{i=1}{\lambda_c^{(i)}} + \frac{n}{T} \leq \lambda_v [{p_{0}\overline{q}_{0} + p_{1}q_{1}} + {p_{1}\overline{q}_{1}}+ ... + p_{n-1}\overline{q}_{n-1} + p_{0}q_{0}]
\end{aligned}
\end{equation}
We have ${\overline{q}_{i} + q_{i}}$ so ${p_{i}\overline{q}_{i} + p_{i}q_{i}} = p_{i}$
\begin{equation}\label{eq:25}
\begin{aligned}
\sum ^{n}_{i=1}{\lambda_c^{(i)}} + \frac{n}{T} \leq \lambda_v ({p_{0} + p_{1} + p_{2}}+ ... + p_{n-1})
\end{aligned}
\end{equation}
We have 
$\sum ^{n-1}_{i=0}{p_{i}} = 1 $ so $\sum ^{n}_{i=1}{\lambda_c^{(i)}} +\frac{n}{T} \leq \lambda_v $

\section{Proof of Lemma \ref{lem3}}\label{app:lem3}
The summation of the inequalities given by (\ref{eq:4}) $\forall~i=\{0,\dots, n\}$ gives the following inequality : 
\begin{equation}
\begin{aligned}
\lambda_v\sum ^{n-1}_{i=0}p_{i} -\lambda_v\sum ^{n-1}_{i=0}p_{i}{q_{i}} +\lambda_v p_{0}q_{0} < C(n \mu_c) + \mu_c
\end{aligned}
\end{equation}
Since $\sum ^{n-1}_{i=0}p_{i}=1$ (because $p_n=0$ as described in Section 2), we get: 
\begin{equation}\label{eq:27}
\begin{aligned}
\lambda_v -\lambda_v\sum ^{n-1}_{i=1}p_{i}{q_{i}}\ignore{ +\lambda_v p_{0}q_{0}} < \mu_c (C n + 1)
\end{aligned}
\end{equation}
In the worst case, all the vehicles will be directed to partially charge before serving, which means that always $q_i = 0$. Therefore, we get:
\begin{equation}\label{eq:28}
\begin{aligned}
C n > \dfrac{\lambda_v}{\mu_c} - 1\;,
\end{aligned}
\end{equation}
which can be re-arranged to be: 
\begin{equation}\label{eq:29}
\begin{aligned}
n > \dfrac{\lambda_v}{C \mu_c} - \dfrac{1}{C}
\end{aligned}
\end{equation}
From equation (\ref{eq:29}) and equation (\ref{eq:9}) we have 
\begin{equation}\label{eq:301}
\begin{aligned}
n > \dfrac{\lambda_v}{C \mu_c} - \dfrac{1}{C} \geq \dfrac{\sum ^{n}_{i=1}{\lambda_c^{(i)}} + \frac{n}{T}}{C \mu_c} - \dfrac{1}{C}
\end{aligned}
\end{equation}
By simplifying equation (\ref{eq:301}) we get
\begin{equation}
\begin{aligned}
n \geq T \dfrac{\sum ^{n}_{i=1}{\lambda_c^{(i)}} - \mu_c}{T C \mu_c - 1}
\end{aligned}
\end{equation}

\section{Proof of Theorem \ref{thm1}} \label{app:thm1}
Applying the KKT conditions to the inequalities constraints of (\ref{eq:11}), we get:
\begin{equation}\label{eq:291}
\begin{aligned}
& & &\alpha_{i}^*( {\lambda_v^*(p_{i-1}{q_{i-1}^*} - p_{i}q_{i}^* - p_{i-1}) + \frac{1}{T}} + \lambda_c^{(i)} ) = 0\\ 
& & & i= 1, \ldots, n-1.\\
& & &\alpha_{n}^* ( {\lambda_v^*( p_{n-1}q_{n-1}^* - p_{0}{q_{0}^*} - p_{n-1}) + \frac{1}{T}} + \lambda_c^{(n)} ) = 0.\\
& & &\beta_{0}^* (\sum ^{n-1}_{i=0}\lambda_v(p_{i} - p_{i}{q_{i}^*}) - C (n \mu_c) + \epsilon_0) = 0.\\
& & &\beta_{1}^*(\lambda_v p_{0}q_{0}^* - \mu_c +\epsilon_1) = 0 \\
& & & \gamma_{i}^*(q_{i}^* - 1) = 0 , \;i = 0, \ldots, n-1.\\
& & &\omega_{i}^*q_{i}^* = 0 , \;i = 0, \ldots, n-1.\\
& & & \omega_{n}^*(\lambda_v^* -( \sum ^{n}_{i=1}{\lambda_c^{(i)}} + \frac{n}{T}) ) = 0.
\end{aligned}
\end{equation}
Likewise, applying the KKT conditions to the Lagrangian function in (\ref{eq:12}), and knowing that the gradient of the Lagrangian function goes to $0$ at the optimal solution, we get the following set of equalities:
\begin{equation}\label{eq:30}
\begin{aligned}
& \lambda_v^* p_{i}(\alpha_{i+1}^* - \alpha_{i}^* - \beta_0^*) = \omega_{i}^* - \gamma_{i}^* , \; i= 1, \ldots, n-1.\\
& \lambda_v^* p_{0}(\alpha_{1}^* - \alpha_{n}^* - \beta_0^* + \beta_1^*) = \omega_{0}^* - \gamma_{0}^*\\
&\sum_{i=1}^{n-1} \alpha_{i}^*(p_{i-1}{q_{i-1}^*} - p_{i}q_{i}^* - p_{i-1}) + \alpha_{n}^* ( p_{n-1}q_{n-1}^* \\ &- p_{0}{q_{0}^*} - p_{n-1}) +  \beta_{0}^* (\sum ^{n-1}_{i=0}(p_{i} - p_{i}{q_{i}^*}))+\beta_{1}^* p_{0}q_{0}^* - \omega_{n}^* + 1= 0
\end{aligned}
\end{equation} 
Knowing that the gradient of the Lagrangian goes to $0$ at the optimal solutions, we get the system of equalities given by (\ref{eq:30}). multiplying the first equality in (\ref{eq:30}) by $q_i^*$ and the second equality by $q_0^*$ and the third equality by $\lambda_v^*$combined with the equalities given by (\ref{eq:291}) gives :
\begin{equation}\label{eq:31}
\begin{aligned}
& \lambda_v^* p_{i}q_i^*(\alpha_{i+1}^* - \alpha_{i}^* - \beta_0^*) = - \gamma_{i}^* , \; i= 1, \ldots, n-1.\\
& \lambda_v^* p_{0}q_0^*(\alpha_{1}^* - \alpha_{n}^* - \beta_0^* + \beta_1^*) = - \gamma_{0}^*\\
& \lambda_v^* - \sum_{i=1}^{n} \alpha_{i}^*(\lambda_c^{(i)} + \frac{1}{T}) +\beta_0^* ( C n \mu_c - \epsilon_0) + \beta_1^* (\mu_c -\epsilon_1)\\ & - \omega_{n}^*( \sum ^{n}_{i=1}{\lambda_c^{(i)}} + \frac{n}{T})= 0
\end{aligned}
\end{equation}
(\ref{eq:31}) Inserted in the fifth equality in (\ref{eq:291}) gives : 
\begin{equation}\label{eq:26}
\begin{aligned}
& & &\lambda_v^* p_{i}(\alpha_{i+1}^* - \alpha_{i}^* - \beta_0^*)(q_{i}^* - 1)q_{i}^* = 0 , \; i= 1, \ldots, n-1.\\
& & &\lambda_v^* p_{0}(\alpha_{1}^* - \alpha_{n}^* - \beta_0^* + \beta_1^*)(q_{0}^* - 1)q_{0}^* = 0\\
& & & \lambda_v^* = \sum_{i=1}^{n} \alpha_{i}^*(\lambda_c^{(i)} + \frac{1}{T}) -\beta_0^* ( C n \mu_c - \epsilon_0) - \beta_1^* (\mu_c -\epsilon_1)\\ & & & + \omega_{n}^*( \sum ^{n}_{i=1}{\lambda_c^{(i)}} + \frac{n}{T})
\end{aligned}
\end{equation}
From (\ref{eq:26}) we have 
$0 < q_{0}^* <1$ only if $\alpha_{i+1}^* - \alpha_{i}^* - \beta_0^* =0$
And $0 < q_{i}^* <1$ only if $\alpha_{1}^* - \alpha_{n}^* - \beta_0^* + \beta_1^*=0$ 
Since $0 \leq q_{i}^* \leq 1$ then these equalities may not always be true 

if $\alpha_{1}^* - \alpha_{n}^* - \beta_0^* + \beta_1^*>0$ and we know that $\gamma_{0}^* \geq 0$ then $\gamma_{0}^* = 0$ which gives $q_{0}^* \ne 1 $ and $q_{0}^* = 0 $.

if $\alpha_{i+1}^* - \alpha_{i}^* - \beta_0^* > 0$ and we know that $\gamma_{i}^* \geq 0$ then $\gamma_{i}^* = 0$ which gives $q_{i}^* \ne 1 $ and $q_{i}^* = 0 $

if $\alpha_{1}^* - \alpha_{n}^* - \beta_0^* + \beta_1^* < 0$ then $\gamma_{0}^* > 0$ (it cannot be 0 because this will contradict with the value of $q_{i}$), which implies that $q_{0}^* = 1 $. 

if $\alpha_{i+1}^* - \alpha_{i}^* - \beta_0^* < 0$ then $\gamma_{i}^* > 0$ (it cannot be 0 because this contradicts with the value of $q_{i}$), which implies that $q_{i}^* = 1 $

We have also from the KKT conditions given by equation in in (\ref{eq:291}) that says either the Lagrangian coefficient is 0 or its the associated inequality is an equality: \\
if $\beta_1^* \ne 0 $ we have $q_{0}^* =\frac{\mu_c}{\lambda_v^* p_0}$\\
if $\alpha_{n}^* \ne 0$ we have $ q_{0}^* =\frac{p_{n-1}{q_{n-1}^* - p_{n-1} }}{p_{0}} + \frac{ \lambda_c^{(n)} + \frac{1}{T}}{\lambda_v p_{0} } \; $\\
if $\alpha_{i}^* \ne 0$ , we have $ q_{i}^* =\frac{p_{i-1}{q_{i-1}^*} -p_{i-1}}{p_{i}} + \frac{\lambda_c^{(i)} + \frac{1}{T}}{\lambda_v p_{i} } $ \\for $i= 1, \ldots, n-1$\\

Otherwise by the Lagrangian relaxation: \\
$q_{i}^* = \zeta_i(\alpha^*,\beta^*,\gamma^*,\lambda_v^*,q^*)$ for $i= 1, \ldots, n-1$\\ 
Where $\zeta_i(\alpha^*,\beta^*,\gamma^*,\lambda_v^*,q^*)$ is the solution that that maximize the function $\underset{\mathbf{q}}\inf~ L(\mathbf{q},\alpha^*,\beta^*,\gamma^*,\lambda_v^*)$\\

Now in order to find the expression of $\lambda_v^*$ we first look at the last equation in (\ref{eq:291}). From there we can say that if 
$\omega_n^* \ne 0$ then $\lambda_v^* = \sum ^{n}_{i=1}{\lambda_c^{(i)}} + \frac{n}{T} $\\

Otherwise from the third equation in (\ref{eq:26})
if $\omega_n^* = 0$ then $ \lambda_v^* =  \sum_{i=1}^{n} \alpha_{i}^*(\lambda_c^{(i)} + \frac{1}{T}) -\beta_0^* ( C n \mu_c - \epsilon_0) - \beta_1^* (\mu_c -\epsilon_1) $

\begin{thebibliography}{1}


\bibitem{ref5}
W. J. Mitchell, C. E. Borroni-Bird, and L. D. Burns, ``Reinventing the
Automobile: Personal Urban Mobility for the 21st Century''. Cambridge,
MA: The MIT Press, 2010.

\bibitem{ref6}
D. Schrank, B. Eisele, and T. Lomax, ``TTIs 2012 Urban Mobility Report,'' \emph{ Texas A\&M Transportation Institute}, Texas, USA.2012.

\bibitem{ref3}
U. N. E. Programme, ``The Emissions Gap Report 2013 - UNEP,''
\emph{Tech. Rep.}, 2013.


\bibitem{ref4}
U. E. P. Agency, ``Greenhouse Gas Equivalencies Calculator,'' \emph{Tech.Rep.}, 2014. [Online]: http://www.epa.gov/cleanenergy/energy-resources/refs.html



\bibitem{ref11}
``IoT And Smart Cars: Changing The World For The Better,'' \emph{Digitalist Magazine}, August
30, 2016. [Online]: http://www.digitalistmag.com/iot/2016/08/30/iot-smart-connected-cars-willchange-world-04422640

\bibitem{ref9}
``Transportation Outlook: 2025 to 2050,'' \emph{Navigant Research, Q2’16}, 2016. [Online]:
http://www.navigantresearch.com/research/transportation-outlook-2025-to-2050.

\bibitem{ref10}
``The Future Is Now: Smart Cars And IoT In Cities,'' \emph{Forbes}, June 13, 2016. [Online]: http://www.forbes.com/sites/pikeresearch/2016/06/13/the-future-is-now-smartcars/63c0a25248c9

\bibitem{ref14}
``Fog Computing and the Internet of Things: Extend the Cloud to Where the Things Are,'' \emph{Cisco White Paper}, 2015. [Online]:http://www.cisco.com/c/dam/en\_us/solutions/trends/iot/docs/computing-overview.pdf

\bibitem{ref1}
R. Zhang, K. Spieser, E. Frazzoli, and M. Pavone, ``Models, Algorithms, and Evaluation for Autonomous Mobility-On-Demand Systems,'' \emph{in Proc. of American Control Conf.}, Chicago, Illinois, 2015.

\bibitem{ref12}
R. Zhang, F. Rossi, and M. Pavone, ``Model Predictive Control of Autonomous Mobility-on-Demand Systems,'' \emph{in Proc. IEEE Conf. on Robotics and Automation}, Stockholm, Sweden, 2016.

\bibitem{ref15}
H. Liang, I. Sharma, W. Zhuang, and K. Bhattacharya,``Plug-in Electric Vehicle Charging Demand Estimation based on Queueing Network Analysis,'' \emph{IEEE Power and Energy Society General Meeting}, 2014.

\bibitem{ref16}
K. Zhang, Y. Mao, S. Leng, Y. Zhang, S. Gjessing, and D.H.K. Tsang, ``Platoon-based Electric Vehicles Charging with Renewable Energy Supply: A Queuing Analytical Model,'' \emph{in Proc. of IEEE International Conference on Communications (ICC’16)}, 2016.

\bibitem{qcqp} 
S. Boyd and J. Park, ``General Heuristics for Nonconvex Quadratically Constrained Quadratic Programming'. \emph{Stanford University}, 2017.
\bibitem{ref13}
S. Belakaria, M. Ammous, S. Sorour, and A. Abdel-Rahim, ``A Multi-Class Dispatching and Charging Scheme for Autonomous Electric Mobility On-Demand,'',\emph{IEEE Vehicular Technology Conference (VTC2017-Fall)}, Toronto, Canada, 2017.
\bibitem{ref17}
M. Ammous, S. Belakaria, S. Sorour, and A. Abdel-Rahim, ``Optimal Routing with In-Route Charging of Mobility-on-Demand Electric Vehicles'',\emph{IEEE Vehicular Technology Conference (VTC2017-Fall)}, Toronto, Canada, 2017.
\bibitem{ref18}
S. Belakaria, M. Ammous, S. Sorour, and A. Abdel-Rahim, ``Optimal Vehicle Dimensioning for Multi-Class Autonomous Electric Mobility On-Demand Systems,'' \emph{ArXiv e-Prints}, 2017.
\end{thebibliography}
\end{document}